\RequirePackage{amsmath}
\documentclass{llncs}						
\usepackage{makeidx}  						
\usepackage{etoolbox}  						
\usepackage[utf8]{inputenc}
\usepackage{amsfonts, amsopn, amsmath}
\usepackage{graphicx}						
\usepackage{enumitem}
\usepackage{xspace}							
\usepackage{multirow}						
\usepackage{sidecap}                 		
\usepackage[textwidth=20mm,disable]{todonotes} 		
\newcommand{\N}[0]{	\mathbb{N}}
\newcommand{\seq}[0]{\subseteq}

\newcommand{\pw}{\ensuremath{\mathop{\mathrm{pw}}}}
\newcommand{\nd}{\ensuremath{\mathop{\mathrm{nd}}}}
\newcommand{\vc}{\ensuremath{\mathop{\mathrm{vc}}}}
\newcommand{\tw}{\ensuremath{\mathop{\mathrm{tw}}}}
\newcommand{\fvs}{\ensuremath{\mathop{\mathrm{fvs}}}}
\newcommand{\FPT}{\ensuremath{\mathsf{FPT}}\xspace}
\newcommand{\XP}{\ensuremath{\mathsf{XP}}\xspace}
\newcommand{\W}[1]{\ensuremath{\mathsf{W[#1]}}}
\newcommand{\NP}{{$\mathsf{NP}$}\xspace}
\newcommand{\NE}{{\mathsf{NE}}}
\newcommand{\E}{{\mathsf{E}}}
\newcommand{\FO}{{$\mathsf{FO}$}\xspace}
\newcommand{\MSO}{{$\mathsf{MSO}$}\xspace}
\newcommand{\MSOt}{{$\mathsf{MSO}_2$}\xspace}
\newcommand{\MSOo}{{$\mathsf{MSO}_1$}\xspace}
\newcommand{\bigO}[1]{\ensuremath{O(#1)}}
\newcommand{\smallo}[1]{\ensuremath{o(#1)}}

\let\phi=\varphi

\def\Land{\mathop{\bigwedge\limits}}
\def\Lor{\mathop{\bigvee\limits}}

\AtEndEnvironment{proof}{
\unskip\unskip~\hfill\qed
}
\newcommand{\mytodo}[2]{\todo[size=\tiny, color=#1!50!white]{#2}\xspace}
\newcommand{\ttcom}[1]{\mytodo{red}{#1}}
\newcommand{\tmcom}[1]{\mytodo{orange}{#1}}

\newcommand{\adj}[2]{\ensuremath{\mathop{\mathrm{adj}}(#1,#2)}}
\newcommand{\inc}[2]{\ensuremath{\mathop{\mathrm{inc}}(#1,#2)}}
\newcommand{\Trule}{\rule{0pt}{3ex}}
\newcommand{\Brule}{\rule[-1.5ex]{0pt}{0pt}}

\newcommand{\prob}[4]{
\begin{definition}[#1]
\vspace{-10pt}
\begin{center}
\begin{tabular} {|llll|}
	\hline
\rule{0pt}{3ex} 
	~&{\bf Input:\enspace}&{\parbox[t]{27em}{#2}}&~\Trule\\
	~&{\bf Question:\enspace}&\parbox[t]{27em}{#3\Brule}&~\\
	\hline
\end{tabular}
\end{center}
\vspace{-3pt}
\end{definition}
}

\begin{document}
\frontmatter          
\pagestyle{headings}  
%
%
%
%
%
\title{Parameterized complexity of fair deletion problems.\thanks{Research was supported by the project GAUK 338216 and by the project SVV-2016-260332.
}}
\titlerunning{Fair deletion problems}  
\author{Tomáš Masařík\inst{1}\thanks{Author was supported by the project CE-ITI P202/12/G061.} \and Tomáš Toufar\inst{2}} 
\authorrunning{Tomáš Masařík and Tomáš Toufar} 
%
\tocauthor{Tomáš Masařík, Tomáš Toufar}
\institute{
Department of Applied Mathematics, Faculty of Mathematics and Physics, Charles University, Prague, Czech Republic, \\
\email{masarik@kam.mff.cuni.cz} \and
Computer Science Institute of Charles University, Faculty of Mathematics and Physics, Charles University, Prague, Czech Republic \\
\email{toufi@iuuk.mff.cuni.cz}.
}
\maketitle              
\begin{abstract}



Deletion problems are those where given a graph $G$ and a graph property $\pi$, the goal is to find a subset of edges such that after its removal the graph $G$ will satisfy the property $\pi$. Typically, we want to minimize the number of elements removed.
In fair deletion problems we change the objective: we minimize the maximum number of deletions in a neighborhood of a single vertex.

We study the parameterized complexity of fair deletion problems with respect to the structural parameters of the  tree-width, the path-width, the size
of a minimum feedback vertex set, the neighborhood diversity, and the size of minimum vertex cover of graph $G$.

We prove the $\W{1}$-hardness of the fair \FO vertex-deletion problem  with respect to the first three parameters combined.
\tmcom{dle mě to chce předělat!!!}Moreover, we show that there is no algorithm for fair \FO vertex-deletion problem running in time $n^{\smallo{\sqrt[3]{k}}}$, where $n$ is the size of the graph and $k$ is the sum
of the first three mentioned parameters, provided that the Exponential Time Hypothesis holds.

On the other hand, we provide an FPT algorithm for the fair \MSO edge-deletion problem parameterized by the size of minimum vertex cover and
an FPT algorithm for the fair \MSO vertex-deletion problem parameterized by the neighborhood diversity.


\end{abstract}
%
	\section{Introduction}
We study the computational complexity of \emph{fair deletion problems}.
Deletion problems are a standard reformulation of some classical problems in combinatorial optimization examined by Yannakakis~\cite{Yannakakis81}. 
For a graph property $\pi$ we can formulate an \emph{edge deletion problem}. That means, given a graph $G=(V,E)$, find the minimum set of edges $F$ that need to be deleted for graph ${G'=(V,E\setminus F})$ to satisfy property $\pi$.
A similar notion holds for the  \emph{vertex deletion problem}.

Many classical problems can be formulated in this way such as {\sc minimum vertex cover, maximum matching} or {\sc minimum feedback arc set}.
For example {\sc minimum vertex cover} is formulated as a vertex deletion problem since we aim to find a minimum set of vertices such that the rest of the graph forms an independent set.
An example of an edge deletion problem is {\sc perfect matching}: we would like to find a minimum edge set such that the resulting graph has all vertices being of degree exactly one. 
Many of such problems are \NP-complete~\cite{Yannakakis78,Watanabe,KriDeo}.


\emph{Fair deletion problems} are such modifications where the cost of the solution should be split such that the cost is not too high for anyone. More formally, the \textsc{fair edge deletion problem} for a given graph $G=(V,E)$ and a property $\pi$ finds a set $F\seq E$ which minimizes the maximum degree of the graph ${G^*=(V,F)}$ where the graph ${G'=(V,E\setminus F)}$ satisfies the property~$\pi$. Fair deletion problems were introduced by Lin and Sahni~\cite{LiSah}.

Minimizing the fair cost arises naturally in many situations, for example in defective coloring~\cite{defcol}. A graph is $(k,d)$-colorable if every vertex can be assigned a color from the set $\{1,\ldots,k\}$ in such a way that every vertex has at most $d$ neighbors of the same color. This problem can be reformulated in terms of fair deletion;
\tmcom{difference between ; and :}
we aim to find a set of edges of maximum degree $d$ such that after its removal the graph can be partitioned into $k$ independent sets.

We focus on fair deletion problems with properties definable in either first order (\FO) or monadic second order (\MSO) logic.
Our work extends the result of Kolman et al.~\cite{Kolman09onfair}. 
They showed an \XP algorithm for a generalization of fair deletion problems definable by \MSOt formula on graphs of bounded tree-width.

We give formal definitions of the problems under consideration in this work.

\prob{\sc Fair \FO edge-deletion}
{An undirected graph $G$, an \FO sentence $\psi$, and a positive integer~$k$.}
{Is there a set $F \subseteq E(G)$ such that $G \setminus F \models \psi$ and for every
vertex $v$ of $G$, the number of edges in $F$ incident with $v$ is at most
$k$?}

Similarly, \textsc{fair vertex deletion problem} finds, for a given graph $G=(V,E)$ and a property $\pi$, the solution which is the minimum of maximum degree of graph ${G[W]}$ where graph ${G[V\setminus W]}$ satisfy property $\pi$. Those problems are \NP-complete for some formulas. For example Lin and Sahni~\cite{LiSah} showed that deciding whether a graph $G$ has a degree one subgraph $H$ such that $G\setminus H$ is a spanning tree is \NP-complete.

\prob{\sc Fair \FO vertex-deletion}
{An undirected graph $G$, an \FO sentence $\psi$, and a positive integer~$k$.}
{Is there a set $W \subseteq V(G)$ such that $G \setminus W \models \psi$ and for every
vertex $v$ of $G$, it holds that $|N(v) \cap W| \leq k$?}

Both problems can be straightforwardly modified for \MSOo or \MSOt.


The following notions are useful when discussing the fair deletion problems.
The \emph{fair cost of a set} $F \subseteq E$ is defined as $\max_{v\in V} |\{ e \in F \mathrel| v \in e \}|$. We refer to the function that assigns each set $F$ its fair cost as the
\emph{fair objective function}. In case of vertex-deletion problems, the \emph{fair cost of a set $W \subseteq V$} is defined as $\max_{v\in V} |N(v) \cap W|$. The \emph{fair objective function} is defined analogously. Whenever we refer to the fair cost or the fair objective function, it should be clear from context whether we mean the edge or the vertex version.

We now describe the generalization of fair deletion problems considered by Kolman et al. The main motivation is that sometimes we want to put additional constraints on the deleted set itself (e.g. \textsc{Connected Vertex Cover}, \textsc{Independent Dominating Set}). However, the framework of deletion problems does not allow that. To overcome this problem, we define the generalized problems as follows.

\prob{\sc Generalized Fair \MSO edge-deletion}
{An undirected graph $G$, an \MSO formula $\psi$ with one free edge-set
variable, and a positive integer $k$.}
{Is there a set $F \subseteq E(G)$ such that $G \models \psi(F)$ and for every
vertex $v$ of $G$, the number of edges in $F$ incident with $v$ is at most
$k$?}

\prob{\sc Generalized Fair \MSO vertex-deletion}
{An undirected graph $G$, an \MSO formula $\psi$ with one free vertex-set variable, and a positive integer $k$.}
{Is there a set $W \subseteq V(G)$ such that $G \models \psi(W)$ and for every vertex $v$ of $G$, it holds that $|N(v) \cap W| \leq k$?}

In this version, the formula $\psi$ can force that $G$ has the desired property after deletion as well as imposing additional constraints on the deleted set itself.

Courcelle and Mosbah~\cite{CourcelleMosbah} introduced a semiring homomorphism framework that can be used
to minimize various functions over all sets satisfying a given \MSO formula. A natural question is
whether this framework can be used to minimize the fair objective function. The answer is no, as we exclude
the possibility of an existence of an \FPT algorithm for parameterization by tree-width under reasonable assumption. Note that there are semirings that capture the fair objective function, but their size is of order $\bigO{n^{\tw{(G)}}}$, so this approach does not lead
to an \FPT algorithm.

\subsection{Our results}
We prove that the \XP algorithm given by Kolman et al.~\cite{Kolman09onfair} is almost optimal under the exponential time hypothesis (ETH) for both the edge and the vertex version. Actually we proved something little bit stronger. We prove the hardness of the classical (weaker) formulation of {\sc fair deletion problems} described in (weaker as well) \FO logic.

\begin{theorem}\label{thm:hardvertex}
If there is an \FPT algorithm for \textsc{Fair \FO vertex-deletion} parameterized by the size of the formula $\psi$,
the pathwidth of $G$, and the size of minimum feedback vertex set of $G$ combined,
then $\FPT = \W{1}$.
Moreover, let $k$ denote $\pw(G)+ \fvs(G)$. If there is an algorithm for
\textsc{Fair \FO vertex-deletion} with running time $f(|\psi|, k) n^{o(\sqrt[3] k)}$,
then Exponential Time Hypothesis fails.
\end{theorem}

\begin{theorem}
\label{thm:edge_deletion_hardness}
If there is an \FPT algorithm for \textsc{Fair \FO edge-deletion} parameterized by the size of the formula $\psi$,
the pathwidth of $G$, and the size of minimum feedback vertex set of $G$ combined,
then $\FPT = \W{1}$.
Moreover, let $k$ denote $\pw(G)+\fvs(G)$.
 If there is an algorithm for
\textsc{Fair \FO edge-deletion} with running time $f(|\psi|, k) n^{o(\sqrt[3] k)}$,
then Exponential Time Hypothesis fails.
\end{theorem}

By a small modification of our proofs we are able to derive tighter ($\sqrt{k}$ instead of $\sqrt[3]{k}$) results using \MSOt logic or \MSOo logic respectively. However, there is still a small gap that has been left open.

\begin{theorem}\label{thm:MSO_hardvertex}
If there is an \FPT algorithm for \textsc{Fair \MSOo vertex-deletion} parameterized by the size of the formula $\psi$,
the pathwidth of $G$, and the size of minimum feedback vertex set of $G$ combined,
then $\FPT = \W{1}$.
Moreover, let $k$ denote $\pw(G)+ \fvs(G)$. If there is an algorithm for
\textsc{Fair \MSOo vertex-deletion} with running time $f(|\psi|, k) n^{o(\sqrt k)}$,
then Exponential Time Hypothesis fails.
\end{theorem}

\begin{theorem}
\label{thm:MSO_edge_deletion_hardness}
If there is an \FPT algorithm for \textsc{Fair \MSOt edge-deletion} parameterized by the size of the formula $\psi$,
the pathwidth of $G$, and the size of minimum feedback vertex set of $G$ combined,
then $\FPT = \W{1}$.
Moreover, let $k$ denote $\pw(G)+\fvs(G)$.
 If there is an algorithm for
\textsc{Fair \MSOt edge-deletion} with running time $f(|\psi|, k) n^{o(\sqrt k)}$,
then Exponential Time Hypothesis fails.
\end{theorem}

On the other hand we show some positive algorithmic results for the generalized version of the problems.

\begin{theorem}
	\label{thm:FPTneighbordiversity}
	\textsc{Generalized Fair \MSOo vertex-deletion} is in \FPT with respect to the neighborhood diversity $\nd(G)$ and the size of the formula $\psi$.
\end{theorem}

We also provide an algorithm for the \MSOt logic (strictly more powerful than \MSOo), however we need a more restrictive parameter because model checking of an \MSOt formula is not even in \XP for cliques unless $\E=\NE$~\cite{Courcelle:00,Lampis:13}.  We consider the size of minimum vertex cover that allows us to attack the edge-deletion problem in \FPT time.  

\begin{theorem}
	\label{thm:FPTvertexCover}
	\textsc{Generalized Fair \MSOt edge-deletion} is in \FPT with respect to the size of minimum vertex cover $\vc(G)$ and the size of the formula $\psi$.
\end{theorem}

	\section{Preliminaries}
Throughout the paper we deal with simple undirected graphs.
For further standard notation in graph theory, we refer to Diestel~\cite{Diestel}.\tmcom{Možná bych Diestela změnil za nějakou jinou knihu, ale obecně mi takové řádky přijdou fajn.}
For terminology in parameterized computational complexity we refer to Downey and Fellows~\cite{df13}.

\subsection{Graph parameters}

We define several graph parameters being used throughout the paper.

\begin{figure}[ht]
  \begin{minipage}[c]{0.5\textwidth}
	\centering
    \includegraphics{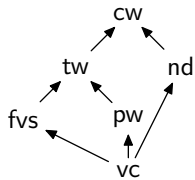}
  \end{minipage}\hfill
  \begin{minipage}[c]{0.5\textwidth}
\caption{Hierarchy of graph parameters. An arrow indicates that a graph parameter upper-bounds the other. Thus, hardness results are implied in direction of arrows and \FPT algorithms are implied in the reverse direction.
}
\end{minipage}
\label{fig:classes}
\end{figure}
We start by definition of \emph{vertex cover} being a set of vertices such that its complement forms an independent set. By $\vc{(G)}$ we denote the size of a smallest such set. This is the strongest of considered parameters and it is not bounded for any natural graph class.

A \emph{feedback vertex set} is a set of vertices whose removal leaves an acyclic graph. Again, by $\fvs{(G)}$ we denote the size of a smallest such set.

Another famous graph parameter is \emph{tree-width} introduced by Bertelé and Brioshi~\cite{bb72}.
\begin{definition}[Tree decomposition]
	A \emph{tree decomposition} of a graph $G$ is a pair $(T,X)$, where ${T=(I,F)}$ is a tree, and $X=\{X_i\mid i\in I\}$ is a family of subsets of $V(G)$ such that:
	\begin{itemize}
		\item the union of all $X_i$, $i\in I$ equals $V$,
		\item for all edges $\{v,w\}\in E$, there exists $i\in I$, such that $v,w\in X_i$ and
		\item for all $v\in V$ the set of nodes $\{i\in I\mid v\in X_i\}$ forms a subtree of $T$.
	\end{itemize}
\end{definition}
The \emph{width} of the tree decomposition is $\max(|X_i|-1)$.
The \emph{tree-width} of a graph $\tw{(G)}$ is the minimum width over all possible tree decompositions of the graph $G$.
The parameter of \emph{path-width} (analogously $\pw{(G)}$) is almost the same except the decomposition need to form a path instead of a general tree.

A less known graph parameter is the \emph{neighborhood diversity} introduced by Lampis~\cite{Lam}.
\begin{definition}[Neighborhood diversity]
	The \emph{neighborhood diversity} of a graph $G$ is denoted by $\nd{(G)}$ and it is the minimum size of a partition of vertices into classes such that all vertices in the same class have the same neighborhood, i.e. ${N(v)\setminus\{v'\}=N(v')\setminus\{v\}}$, whenever
	$v,v'$ are in the same class.
\end{definition}
It can be easily verified that every class of neighborhood diversity is either a clique or an independent set.
Moreover, for every two distinct classes $C$ and $C'$, either every vertex in $C$ is adjacent to every vertex in $C'$,
or there is no edge between them. If classes $C$ and $C'$ are connected by edges,
we refer to such classes as \emph{adjacent}. 
\tmcom{ TODO?? přidat event: generalizes vc(G) and incomparable with tw }

\subsection{Parameterized problems and Exponential Time Hypothesis}



\begin{definition}[Parameterized language]
Let $\Sigma$ be a finite alphabet.
A \emph{parameterized language} $L \subseteq \Sigma^\ast \times \N$ set of pairs $(x, k)$ where $x$ is a finite word over $\Sigma$ and $k$ is a nonnegative integer.
\end{definition}
We say that an algorithm for a parameterized problem $L$ is an \emph{\FPT algorithm} if there exist
a constant $c$ and a computable function $f$ such that the running time for input $(x,k)$ is $f(k)|x|^c$ and the algorithm accepts
$(x,k)$ if and only if $(x,k) \in L$.

A standard tool for showing nonexistence of an \FPT algorithm is \W{1}-hardness (assuming $\FPT \neq \W{1}$).
For the definition of \W{1} class and the notion of \W{1}-hardness, we refer the reader to~\cite{df13}.

A stronger assumption than $\FPT \neq \W{1}$ that can be used to obtain hardness results is
the Exponential Time Hypothesis (ETH for short). It is a complexity theoretic assumption introduced by Impagliazzo, Paturi and Zane~\cite{IPZ01:ETH}.
We follow a survey on the topic of lower bounds obtained from ETH by Lokshtanov, Marx, and Saurabh~\cite{LMS11:ETHLowerBoundsSurvey}, which contains more details on this topic.

The hypothesis states that there is no subexponential time algorithm for {\sc 3-SAT} if we measure the time complexity by the number of variables in the input formula, denoted by $n.$

\vskip .2cm
\begin{minipage}[c]{.9\textwidth}
{\bf Exponential Time Hypothesis (ETH)~\cite{IPZ01:ETH}}
There is a positive real $s$ such that {\sc 3-SAT} with parameter $n$ cannot be solved in time ${2^{sn}(n+m)^{\bigO{1}}}.$
\end{minipage}
\vskip .2cm


\begin{definition}[Standard parameterized reduction]\label{def:reduction}
	We say that parameterized language $L$ reduces to parameterized language $L'$ by a \emph{standard parameterized reduction} if there are functions ${f,g\colon\N\to\N}$
and ${h\colon\Sigma^*\times\N\to\Sigma^*}$ such that
		function $h$ is computable in time $g(k) |x|^c$ for a constant $c$, and
		$(x,k)\in L$ if and only if $(h(x,k),f(k))\in L'$.
\end{definition}

For preserving bounds obtained from the ETH, the asymptotic growth of the function $f$ need to be as slow as possible.

\subsection{Logic systems}
\label{subsec:logic_systems}
\tmcom{možná bych zde vysvětlil co myslíme tím size of formula.}
We heavily use graph properties that can be expressed in certain types of logical systems.
In the paper it is \emph{Monadic second-order logic} (\MSO) where monadic means that we allow quantification over sets (of vertices and/or edges). In \emph{first order logic} (\FO) there are no set variables at all.\tmcom{přeformulovat}

We distinguish \MSOt and \MSOo. In \MSOo quantification only over sets of vertices is allowed and we can use the predicate of adjacency \adj{u}{v} returning true whenever there is an edge between vertices $u$ and $v$.
In \MSOt we can additionally quantify over sets of edges and we can use the predicate of incidence \inc{v}{e} returning true whenever a vertex $v$ belongs to an edge $e$.

It is known that \MSOt is strictly more powerful than \MSOo. For example, the property that a graph is Hamiltonian is expressible in \MSOt but not in \MSOo \cite{LibkinFMT}. 

Note that in \MSOo it is easy to describe several complex graph properties like being connected or having a vertex of a constant degree.

	\section{Hardness results}

In this section, we prove hardness of \textsc{Fair \FO vertex-deletion} by exhibiting a reduction from \textsc{Equitable 3-coloring}.

\prob{Equitable 3-coloring}
{An undirected graph $G$.}
{Is there a proper coloring of vertices of $G$ by at most $3$ colors such that the size of any two color classes differ by at most one?}

The following result was proven implicitly in~\cite{eq_coloring}.

\begin{theorem}
	\label{thm:eq_col_hardness}
	\textsc{Equitable 3-coloring} is $\W{1}$-hard with respect to $\pw(G)$ and $\fvs(G)$ combined. Moreover, if there exists an algorithm for \textsc{Equitable 3-color\-ing} running in time $f(k)n^{o(\sqrt[3]k)}$, where $k$ is $\pw(G) + \fvs(G)$, then the Exponential Time Hypothesis fails.
\end{theorem}

The proof in~\cite{eq_coloring} relies on a reduction from \textsc{Multicolored Clique}~\cite{DBLP:journals/tcs/FellowsHRV09} to \textsc{Equitable coloring}. The reduction transforms an instance of \textsc{Multicolored clique} of parameter $k$ into an \textsc{Equitable coloring} instance of path-width and feedback vertex size at most $\bigO{k}$ (though only tree-width is explicitly stated in the paper). Algorithm  for \textsc{Equitable coloring} running in time $f(k)n^{o(\sqrt[3] k)}$ would lead to an algorithm for \textsc{Multicolored Clique} running in time $f(k)n^{o(k)}$. It was shown by Lokshtanov, Marx, and Saurabh~\cite{LMS11:ETHLowerBoundsSurvey} that such algorithm does not exist unless ETH fails.

We now describe the idea behind the reduction from \textsc{Equitable 3-coloring} to \textsc{Fair \FO vertex-deletion}. Let us denote by $n$ the number of vertices of $G$ and assume that $3$ divides $n$. The vertices of $G$ are referred to as \emph{original vertices}. First, we add three vertices called \emph{class vertices}, each of them corresponds to a particular color class. Then we add edge between every class vertex and every original vertex and subdivide each such edge. The vertices subdividing those edges are called \emph{selector vertices}.

We can encode the partition of $V(G)$ by deleting vertices in the following way: if $v$ is an original vertex and $c$ is a class vertex, by deleting the selector vertex between $v$ and $c$ we say\ttcom{tohle je pekne dementni vyraz, ale o pul druhe nic lepsiho nemam} that vertex $v$ \emph{belongs} to the class represented by $c$. If we ensure that the set is deleted in such a way that every vertex belongs to exactly one class, we obtain a partition of $V(G)$.

The equitability of the partition will be handled by the fair objective function. Note that if we delete a subset $W$ of selector vertices that encodes a partition then $|W| = n$. Those $n$ vertices are adjacent to $3$ class vertices, so the best possible fair cost is $n/3$ and thus a solution of the fair cost $n/3$ corresponds to an equitable partition.

Of course, not every subset $W$ of vertices of our new graph encodes a partition. Therefore, the formula we are trying to satisfy must ensure that:
\begin{itemize}
	\item every original vertex belongs to exactly one class,
	\item no original or class vertex was deleted,
	\item every class is an independent set.
\end{itemize}
However, the described reduction is too naive to achieve those goals; we need to slightly adjust the reduction.
Let us now describe the reduction formally:
\begin{proof}[of Theorem~\ref{thm:hardvertex}]
Let $G$ be a graph on $n$ vertices. We can  assume without loss of generality (by addition of isolated vertices.) that $3$ divides $n$ and $n\ge 6$.

First we describe how to construct the reduction. All vertices of $G$ will be referred to as \emph{original vertices}. We add three vertices called \emph{class vertices} and connect every original vertex with every class vertex by an edge. We subdivide each such edge once; the vertices subdividing those edges are called \emph{selector vertices}. Finally, for every original vertex $v$, we add $n$ new vertices called \emph{dangling vertices} and connect each of them by an edge to $v$. We denote the graph obtained in this way as $G'$. For a schema of the reduction, see Figure~\ref{fig:reduction_figure}.

\begin{figure}
	\begin{center}
	\includegraphics{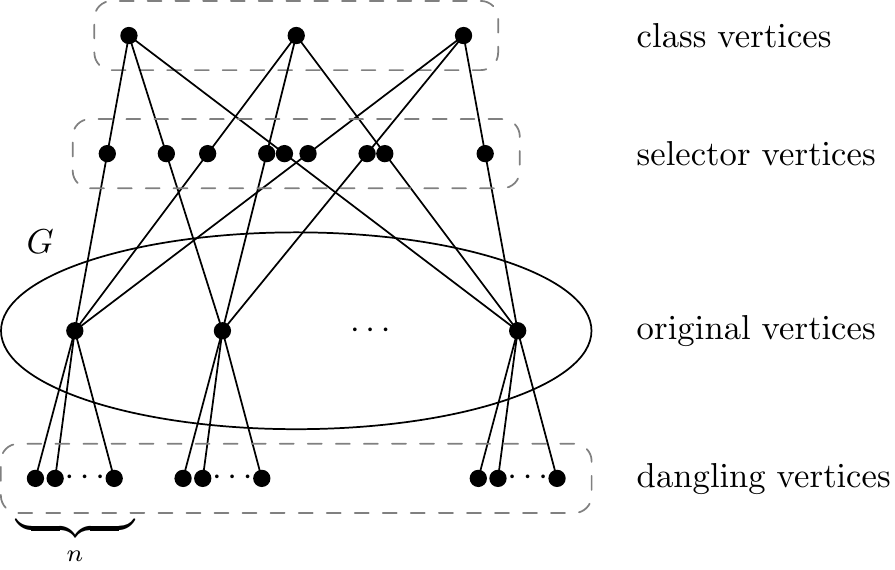}
	\end{center}
	\caption{The schema of the reduction}
	\label{fig:reduction_figure}
\end{figure}

Now, we wish to find a set $W\seq V(G')$ such that it encodes an equitable 3-coloring of a graph $G$. The set is described by the following \FO formula $eq\_3\_col$ imposed on a graph $G\setminus W$. We claim that  whenever this set satisfy following claims it encodes an equitable 3-coloring. A set $W$ can contain only selector vertices and some dangling vertices (but those do not affect the coloring). For each vertex $v$ of a graph there can be only one selector vertex in the set $W$ and that vertex has only one class vertex as a neighbor. That vertex determine the color 
of $v$.\tmcom{Mě se to stejně více hodí až dozadu.}

We use the following shorthand $\exists_{=k}$ meaning there are exactly $k$ distinct elements satisfying a given predicate:
\begin{multline*}
(\exists_{=k} w)(pred(w)) \equiv (\exists v_1, \ldots, v_k)\bigg(
	\Land_{i=1}^k pred(v_i) \land \Land_{1 \leq i < j \leq k}(v_i \neq v_j)\\
	\land (\forall v')\Big(pred(v') \rightarrow \Lor_{i=1}^k (v' = v_i)\Big)
	\bigg)
\end{multline*}
The building blocks for the formula are as follows:
{
\allowdisplaybreaks
\begin{align*}
	isol(v) &\equiv (\forall w)(\lnot adj(v,w)) \\
	dangling(v) &\equiv (\exists w)\big(adj(v,w) \land (\forall w')(adj(v,w') \rightarrow w = w')\big) \\
	original(v) &\equiv (\exists w)(dangling(w) \land adj(v,w)) \\
	selector(v) &\equiv (\exists_{=2} w)(adj(v,w))\\
	class(v) &\equiv \lnot orig(v) \land \lnot selector(v) \land \lnot dangling(v) \\
	belongs\_to(v,a) &\equiv original(v) \land class(a) \land \lnot(\exists w)(adj(v,w) \land adj(w,a)) \\
	same\_class(v,w) &\equiv original(v) \land original(w) \\
	 &\quad \land (\exists a)(class(a) \land belongs\_to(v,a) \land belongs\_to(w,a)) \\
	valid\_deletion &\equiv (\forall v)(\lnot isol(v)) \\ & \quad \land (\forall v)\big(original(v) \rightarrow (\exists_{=1} c)(belongs\_to(v,c))\big)\\
	eq\_3\_col &\equiv valid\_deletion \land (\forall v,w)(same\_class(v,w) \rightarrow \lnot adj(v,w)) \\
\end{align*}
}

The described reduction maps an instance $G$ of an \textsc{Equitable coloring} into an instance $(G', eq\_3\_col, n/3)$ of \textsc{Fair \FO vertex-deletion}.

We claim that there exists a set $W \subseteq V(G')$ of the fair cost at most $n/3$ if and only if $G$ admits an equitable 3-coloring.

If we have an equitable $3$-coloring of $G$ then it is easy to see that the set $W \subseteq V(G')$ corresponding to a partition into color classes has the fair cost exactly $n/3$ and it is straightforward to check that $G' \setminus W \models eq\_3\_col$.

For the other implication we prove that if we delete a subset $W \subseteq V(G')$ of the fair cost at most $n/3$, and the formula $valid\_deletion$ is true, then we obtained an equitable 3-coloring of a graph $G$. To get there we made a few basic claims. 

\if false
\todo[inline]{
\begin{itemize}
\item no original vertex was deleted,
\item if $w$ has degree one in $G' \setminus W$, then its only neighbor is the original vertex,
\item formula $original$ correctly recognizes original vertices,
\item If $v$ is a dangling vertex then $dangling(v)$ is true,
\item if $class(v)$ is true if and only if $v$ is a class vertex,
\item no class vertex was deleted.
\end{itemize}
}
\fi

\emph{Claim 1: no original vertex was deleted:}
Suppose for the contradiction that original vertex $v$ was deleted. If we kept at least one of the dangling vertices attached to $v$, but this vertex is now isolated and formula $valid\_deletion$ is not true. On the other hand if we delete all dangling vertices that were attached to $v$, our deleted set has fair cost at least $n$.

\emph{Claim 2: if $w$ has degree one in $G' \setminus W$, then its only neighbor is an original vertex:}
If $w$ is dangling, then its only neighbor is original vertex by the construction of $G'$. Suppose that $w$ has degree one in $G' \setminus W$ but is not dangling. Since both class and original vertices have degree at least $n$ in $G'$, we cannot bring them down to degree one without exceeding the fair cost limit $n/3$. This leaves the only possibility that $w$ is a selector and exactly one of its two neighbors is in the deleted set $W$. By Claim 1, the deleted neighbor must have been a  class vertex so the only remaining neighbor of $w$ in $G' \setminus W$ is an original vertex.

\emph{Claim 3: the formula $original$ correctly recognizes original vertices:}
If $v$ is original, then at least one of its dangling neighbors is not in $W$, otherwise we would exceed the fair cost. In this case the formula $original(v)$ is true. The other direction ($original(v)$ is true implies $v$ is original) is proved by Claim 2.

\emph{Claim 4: if $v$ is a dangling vertex such that $v \notin W$ then $dangling(v)$ is true:} By Claim 1, we cannot delete the only neighbor of $v$, which means $v$ has exactly one neighbor and so $dangling(v)$ is true.

\emph{Claim 5: the formula $class(v)$ is true if and only if $v$ is a class vertex that was not deleted:} Suppose that $v \notin W$ is a class vertex. It cannot have neighbor of degree one in $G' \setminus W$, because that would mean that an original vertex was deleted which violates Claim 1. This means that $original(v)$ is false. Moreover, we cannot decrease the degree of $v$ to two or less by deleting at most $n/3$ neighbors of $v$, so $dangling(v)$ and $selector(v)$ are false too. But then $class(v)$ is true.

For the other direction suppose that $v$ is not a class vertex. If it is original or dangling, then $original(v)$ or $dangling(v)$ is true (by Claim 3 or Claim 4) and hence $class(v)$ is false. If $v$ is a selector then either none of its neighbors were deleted, $v$ has degree two in $G' \setminus W$ and $selector(v)$ is true, or its class neighbor was deleted, $v$ has degree one in $G' \setminus W$ and $dangling(v)$ is true. Either way, $class(v)$ is false as required.

\emph{Claim 6: no class vertex was deleted:} since $valid\_deletion$ is true, we know that for every original vertex $v$ there is exactly one class vertex $c$ such that there is no path of length two between $v$ and $c$ (in other words, the selector vertex that was on the unique path of length two between $v$ and $c$ was deleted). Suppose for contradiction that one of the class vertices was deleted; then by Claim 5 we have at most two class vertices. But the $valid\_deletion$ formula implies that at least $n$ selector vertices were deleted. By pigeonhole principle, one of the class vertices has at least $n/2$ deleted neighbors which means the fair cost is greater than $n/3$, a contradiction.

The chain of claims we just proved guarantees that the deleted set $W$ indeed obeys the rules we required and corresponds to a partition (though we might have deleted a small number of dangling vertices, this does not affect the partition in any way). In order to meet the fair cost limit, each class of the partition must have at most $n/3$ vertices and since no original vertex was deleted, it has exactly $n/3$ vertices. Now it is easy to see that the formula $eq\_3\_col$ forces that each class of the partition is independent and so the graph $G$ has an equitable $3$-coloring.

Let us now discuss the parameters and the size of the \textsc{Fair \FO vertex-deletion} instance. If $G$ has a feedback vertex set $S$ of size $k$, then the union of $S$ with the set of class vertices is a feedback vertex set of $G'$. Therefore, $\fvs(G') \leq \fvs(G) + 3$. To bound the path-width, observe that after deletion of the class vertices we are left with $G$ with $\bigO{n^2}$ added vertices of degree one; the addition of degree one vertices to the original vertices can increase the path-width by at most one and so we have $\pw(G') \leq \pw(G) + 4$. Moreover it is clear that the size of instance is of size $\bigO{n^2}$. It is obvious that the reduction can be carried out in polynomial time.
\end{proof}

Let us mention that if we are allowed to use \MSO formulas, we are actually able to reduce any equitable partition problem to fair vertex deletion. This allows us to reduce for example \textsc{Equitable connected partition} to \textsc{Fair \MSO vertex-deletion} which in turn allows us to prove Theorem~\ref{thm:MSO_hardvertex}.

\prob{Equitable connected partition}
{An undirected graph $G$, a positive integer $r$}
{Is there a partition of $V(G)$ into $r$ sets such that each of them induces a connected graph and the sizes of every two sets differ by at most one?}

Enciso et al.~\cite{eq_conn_part} showed that \textsc{Equitable Connected Partition} is \W{1}-hard for combined parameterization by $\fvs(G)$, $\pw(G)$, and the number of partitions $r$. The part that $f(k)n^{o(\sqrt k)}$ algorithm would refute ETH is again contained only implicitly; the proof reduces an instance of \textsc{Multicolored clique} of parameter $k$ to an instance of \textsc{Equitable connected partition} of parameter $\bigO{k^2}$.

Our reduction can be easily adapted to $r$ parts (we just add $r$ class vertices and we set the fair cost limit to $n / r$). We define the formula $eq\_conn$ as follows.

\begin{align*}
	class\_set(W) &\equiv (\exists v \in W) \land (\forall v,w \in W)( same\_class(v,w)) \\
		 &\quad \land (\forall w \in W, z \notin W)(\lnot same\_class(w,z)) \\
	eq\_conn &\equiv (\forall W)(class\_set(W) \rightarrow connected(W)) \\
\end{align*}

By the same argument as in the proof of Theorem~\ref{thm:hardvertex}, we can show that there exists $W \subseteq V$ of fair cost at most $n/r$ such that $G' \setminus W \models eq\_conn$ if and only if $G$ admits an equitable connected partition.

\smallskip
\noindent\emph{Sketch of proof of Theorem~\ref{thm:edge_deletion_hardness}:}  We do not present the complete proof, as the critical parts are the same as in proof of Theorem~\ref{thm:hardvertex}.
The reduction follows the same idea as before: we add three class vertices and connect each class vertex to each original vertex by an edge. This time, we do not subdivide the edges, as the partition is encoded by deleting the edges.

The protection against tampering with the original graph has to be done in slightly different way: in this case, we add $n/3 + 1$ dangling vertices of degree one to each original vertex.
Note that if we delete a set $F \subseteq E(G)$ of fair cost at most $n/3$, at least one of the added edges from every original vertex survives the deletion, so we can recognize the original vertices by having at least one neighbor of degree one.
In our formula, we require that each vertex has at most two neighbors of degree one. This forces us to delete all of those added edges except two. Since at least one edge from the original vertex must be deleted to encode a partition, by deleting an edge of the original graph $G$ we would exceed the fair cost limit $n/3$.

For the edge-deletion the formula $eq\_3\_col$ is built as follows.

\allowdisplaybreaks
\begin{align*}
	dangling(v) &\equiv (\exists w)\big(adj(v,w) \land (\forall w')(adj(v,w') \rightarrow w = w')\big) \\
	original(v) &\equiv (\exists w)(dangling(w) \land adj(v,w)) \\
	class(v) &\equiv \lnot orig(v) \land \lnot dangling(v) \\
	belongs\_to(v,a) &\equiv original(v) \land class(a) \land \lnot adj(v,a) \\
	same\_class(v,w) &\equiv original(v) \land original(w) \\
	 &\quad \land (\exists a)(class(a) \land belongs\_to(v,a) \land belongs\_to(w,a)) \\
	valid\_deletion &\equiv (\forall v)( \exists_{\leq 2} w)(adj(v,w) \land dangling(w)) \\ & \quad \land (\forall v)\big(original(v) \rightarrow (\exists_{=1} c)(belongs\_to(v,c))\big)\\
	eq\_3\_col &\equiv valid\_deletion \land (\forall v,w)(same\_class(v,w) \rightarrow \lnot adj(v,w)) \\
\end{align*}

The complete proof of correctness is omitted due to space considerations, however, it is almost exactly the same as in the proof of Theorem~\ref{thm:hardvertex}.\hfill\qed

The transition between the \FO case and the \MSO case of edge-deletion (Theorem~\ref{thm:MSO_edge_deletion_hardness}) is done in exactly the same way as before.

	\section{FPT algorithms}

We now turn our attention to FPT algorithms for fair deletion problems. 

\subsection{FPT algorithm for parameterization by neighborhood diversity}


\begin{definition}
	Let $G = (V,E)$ be a graph of neighborhood diversity $k$ and let $N_1,\ldots,N_k$
denote its classes of neighborhood diversity. 
A \emph{shape of a set $X \subseteq V$ in $G$} is a $k$-tuple $s=(s_1,\ldots,s_k)$, where $s_i = |X \cap N_i|$.

We denote by $\overline s$ the \emph{complementary shape to $s$}, which is defined as the shape of $V \setminus X$,
i.e. $\overline{s} = (|N_1| - s_1, \ldots, |N_k| - s_k)$.
\end{definition}

\begin{proposition} 
\label{prop:property_depends_on_shape}
Let $G = (V,E)$ be a graph, $\pi$ a property of a set of vertices, and let $X,Y \subseteq V$ be two sets of the same shape in $G$. Then $X$ satisfies $\pi$ if and only if $Y$ satisfies $\pi$.
\end{proposition}
\begin{proof}
	Clearly, we can construct an automorphism of $G$ that maps $X$ to $Y$. 
\end{proof}


\begin{definition} 
Let $r$ be a non-negative integer and let $(s_1, \ldots, s_k)$, $(t_1, \ldots, t_k)$ be two shapes. The
shapes are \emph{$r$-equivalent}, if for every $i$:
\begin{itemize}
	\item $s_i = t_i$, or
	\item both $s_i$, $t_i$ are strictly greater than $r$,
\end{itemize}
and the same condition hold for the complementary shapes $\overline s$, $\overline t$.
\end{definition}

The following proposition gives a bound on the number of $r$-nonequivalent shapes.
\begin{proposition}
\label{prop:num_of_noneq_shapes}
	For any graph $G$ of neighborhood diversity $k$, the number
	of $r$-nonequivalent shapes is at most $(2r+3)^k$.
\end{proposition}
\emph{Proof.}
	We show that for every $i$, there are at most $(2r+3)$ choices of $s_i$.
	This holds trivially if $|N_i| \leq 2r+3$. Otherwise we have following $2r+3$ choices:

	\begin{itemize}
		\item $s_i = k$ and $\overline{s_i} > r$ for $k = 0,1,\ldots,r$, or
		\item both $s_i, \overline{s_i} > r$, or
		\item $s_i > r$ and $\overline{s_i} = k$ for $k = 0,1,\ldots,r$.
	\end{itemize}
	\vskip-20pt \hfill\qed

The next lemma states that the fair cost of a set can be computed from its shape in a straightforward manner.
Before we state it, let us introduce some auxiliary notation.

If a graph $G$ of neighborhood diversity $k$ has classes of neighborhood diversity $N_1,\ldots,N_k$,
we write $i \sim j$ if the classes $N_i$ and $N_j$ are adjacent. If the class $N_i$ is a clique,
we set $i \sim i$.
Moreover, we set $\eta_i =1$ if the class $N_i$ is a clique and $\eta_i = 0$ if it is an independent set.
The classes of size one are treated as cliques for this purpose.

\begin{lemma}
\label{lem:fair_cost_from_shape}
Let $G = (V,E)$ be a graph of neighborhood diversity $k$ and let $N_i$ be its classes of neighborhood diversity.
Moreover, let $X \subseteq V$ be a set of shape $s$. Then
the fair vertex cost of $X$ is
$$ \max_{i} \bigg(\Big(\sum_{j:i\sim j} s_j\Big)- \eta_{i}\bigg).$$ 
\end{lemma}
\begin{proof}
It is straightforward to check that vertex $v \in N_i$ has exactly
$\sum_{j:i\sim j} s_j - \eta_{i}$ neighbors in $X$. 
\end{proof}

Our main tool is a reformulation of Lemma~5 from \cite{Lam}:
\begin{lemma}
\label{lem:formula_and_large_shape}
Let $\psi$ be an \MSOo formula with one free vertex-set variable, $q_E$ vertex element quantifiers,
and $q_S$ vertex set quantifiers. Let $r = 2^{q_S}q_E$. If $G = (V,E)$ is a graph
of neighborhood diversity $k$ and
$X,Y \subseteq V$ are two sets such that their shapes are $r$-equivalent, then $G \models \psi(X)$ if and only if
$G \models \psi(Y)$.
\end{lemma}

The last result required is the \MSOo model checking for graphs of bounded neighborhood diversity~\cite{Lam}:
\begin{theorem}
\label{thm:nd_MSO_model_checking}
Let $\psi$ be an \MSOo formula with one free vertex-set variable. There exists an \FPT algorithm
that given a graph $G = (V,E)$ of neighborhood diversity $k$  and a set $X\subseteq V$ decides whether $G \models \psi(X)$.
The running time of the algorithm is $f(k,|\psi|)n^{\bigO{1}}$.
\end{theorem}

We now have all the tools required to prove Theorem~\ref{thm:FPTneighbordiversity}.
\begin{proof}[Proof of Theorem~\ref{thm:FPTneighbordiversity}]
Let $\psi$ be an \MSOo formula in the input of \textsc{Fair \MSOo vertex-deletion}. Denote by $q_S$ the number
of vertex-set quantifiers in $\psi$, by $q_E$ the number of vertex-element quantifiers in $\psi$, and set $r = 2^{q_S}q_E$.

By Proposition~\ref{prop:property_depends_on_shape}, the validity of $\psi(X)$ depends only on the shape of $X$.
Let us abuse notation slightly and write $G \models \psi(s)$ when ``$X$ has shape $s$'' implies $G\models \psi(X)$. Similarly, Lemma~\ref{lem:fair_cost_from_shape} allows us to refer to the fair cost of a shape $s$.

From Lemma~\ref{lem:formula_and_large_shape} it follows that the validity of $\psi(s)$ does not depend on the choice
of an $r$-equivalence class representative. The fair cost is not same for all $r$-equivalent shapes, but since the fair cost is monotone in $s$, we can easily find the representative of the minimal fair cost.

Suppose we have to decide if there is a set of a fair cost at most $\ell$. The algorithm will proceed as follows:
For each class of $r$-equivalent shapes, pick a shape $s$ of the minimal cost, if the fair cost is at most $\ell$ and $G \models \psi(s)$,
output \texttt{true}, if no such shape is found throughout the run, output \texttt{false}.

By the previous claims, the algorithm is correct. Let us turn our attention to the running time. The number
of shapes is at most $(2r+3)^k$ by Proposition~\ref{prop:num_of_noneq_shapes}, and so it is bounded
by $f(|\psi|,k)$ for some function $f$. The \MSOo model checking runs in time $f'(|\psi|,k)n^{\bigO{1}}$ by Theorem~\ref{thm:nd_MSO_model_checking},
so the total running time is $f(|\psi|,k)f'(|\psi|,k)n^{\bigO{1}}$, so the described algorithm
is in \FPT.
\end{proof}


\subsection{FPT algorithm for  parameterization by vertex cover}
The FPT algorithm for parameterization by the size of minimum vertex cover uses the same idea.
We use the fact that every \MSOt formula can be translated to \MSOo formula --- roughly speaking,
every edge-set variable is replaced by $\vc{(G)}$ vertex-set variables.

We only sketch translation from \MSOt to \MSOo, for the proof we refer
the reader to Lemma~6 in~\cite{Lam}. Let $G = (V,E)$ be a graph
with vertex cover ${C = \{v_1,\ldots,v_k\}}$ and $F\subseteq E$ a set of edges.
We construct vertex sets $U_1,\ldots,U_k$ in the following way: if $w$ is
a vertex such that an edge in $F$ connects $w$ with $v_i$, we put $w$ into $U_i$.
It is easy to see that the sets $U_1,\ldots,U_k$ together with the vertex cover
$v_1,\ldots, v_k$ describe the set $F$.

In this way, we reduce the problem of finding a set $F$ to finding $k$-tuple
of sets $(U_1, \ldots, U_k)$. We can define shapes and classes of $r$-equivalence
in an analogous way as we did in previous section. Since the number of $r$-equivalence classes defined in this way is still bounded, we can use essentially the same algorithm:
for each class of $r$-equivalence, run a model checking on a representative of this class.
From those representatives that satisfy $\psi$, we choose the one with best fair cost.

The translation from set of edges into $k$ sets of vertices is captured by the following definition.
\begin{definition}
	Let $G = (V,E)$ be a graph with vertex cover $v_1,\ldots,v_k$. For a set $F\subseteq E$, we define
	\emph{the signature of $F$ with respect to $v_1,\ldots,v_k$} as the $k$-tuple ${\cal U} = (U_1,\ldots,U_k)$,
	where $U_i = \{ w \in V \mid \{w,v_i\} \in F\}$.
	We refer to it simply as \emph{the signature} of $F$ and denote it by $S(F)$ if the vertex cover is clear from the context.
\end{definition}

In the original problem, we had an \MSOt formula $\psi_2$ with one free edge-set variable.
By the translation, we obtain an \MSOo formula $\psi$ with $k$ free vertex-set variables
and $k$ free vertex-element variables (the vertex-element variables will describe the vertex
cover; the formula need to have access to a vertex cover and it will be useful to fix one throughout the whole run of the 
algorithm).

We start by finding a vertex cover $v_1,\ldots,v_k$ (this can be solved by an \FPT algorithm \cite{df13}).
We now want to find the sets $U_1,\ldots,U_k$ such that: $${G \models \psi(v_1,\ldots,v_k,U_1,\ldots,U_k)}.$$
To find such $k$-tuple of sets, we need to extend the notion of shapes to signatures.
\begin{definition}
	Let $G = (V,E)$ be a graph with vertex cover $v_1,\ldots,v_k$, and let ${\cal U} = (U_1,\ldots,U_k)$
	be a collection of $k$ subsets of $V$.
	Denote by $N_1,\ldots,N_\ell$ the classes of neighborhood diversity of $G$.
	For $j \in \{1,\ldots,\ell\}$ and $I \subseteq \{1 \ldots k\}$, denote by $\overline I$ the
	set $\{1,\ldots,k\} \setminus I$. Furthermore, we define $S_{\cal U}(j,I)$ as
	$$ S_{\cal U}(j,I) = \bigg|N_j \cap \bigcap_{i \in I} U_i \cap \bigcap_{i\in \overline I} (V \setminus U_i) \bigg|.$$

	The mapping $S_{\cal U}$ is called \emph{the shape of a signature $\cal U$}.
\end{definition}

The shapes defined in this way have properties similar to those defined for neighborhood diversity; we only state those
properties without proofs.

\begin{definition}
	Two shapes $S$, $S'$ are $r$-equivalent if for every $j \in \{1,\ldots,k\}$, $I \subseteq \{1,\ldots,k\}$ it holds
	that
\begin{itemize}
	\item $S(j,I) = S'(j,I)$, or
	\item both $S(j,I)$, $S'(j,I)$ are strictly greater than $r$.
\end{itemize}
\end{definition}

As in the neighborhood diversity case, the number of $r$-nonequivalent shapes is bounded by a function of $r$ and $k$.
\begin{proposition}
\label{prop:extShapeCount}
	Let $G = (V,E)$ be a graph with vertex cover $v_1,\ldots,v_k$ and denote by $\ell$
	the neighborhood diversity of $G$.
	The number of $r$-nonequivalent shapes is at most $(2r+3)^{\ell 2^k}$.
\end{proposition}

We now state corresponding variants of Lemma~\ref{lem:fair_cost_from_shape} and Lemma~\ref{lem:formula_and_large_shape}.
\begin{lemma}
	Let $G = (V,E)$ be a graph with a vertex cover $v_1,\ldots, v_k$ and let $F \subseteq E$.

The number of edges in $F$ incident to $v_i$ is $|U_i|$. If $w$ is a vertex different from $v_1,\ldots,v_k$,
then the number of edges in $F$ incident to $w$ is $|\{ i \mid w \in U_i \}|$.

Those quantities (and therefore the fair cost of $F$) can be determined from the shape of $S(F)$.
\end{lemma}

\begin{lemma}
	Let $G = (V,E)$ be a graph with a vertex cover $v_1,\ldots, v_k$, let
	$\psi$ be an \MSOo formula with $k$ free vertex-element variables and $k$ free vertex-set variables, and let ${\cal U} = (U_1,\ldots,U_k)$,
	${\cal W} = (W_1, \ldots, W_k)$ be two signatures. If the shapes of $\cal U$ and $\cal W$ are $r$-equivalent,
	then $G \models \psi(v_1,\ldots,v_k, U_1, \ldots, U_k)$ if and only if $G \models \psi(v_1,\ldots,v_k, W_1,\ldots,W_k)$.
\end{lemma}

\begin{proof}[Proof of Theorem~\ref{thm:FPTvertexCover}]
The algorithm goes as follows:
\begin{itemize}
	\item we translate the \MSOt formula $\psi_2$ with one free edge-set variable to
		the \MSOo formula $\psi$ with $k$ vertex-element variables and $k$ vertex-set variables.
	\item We find a vertex cover $c_1,\ldots, c_k$.
	\item For each class of $r$-equivalent shapes, we pick the one achieving the minimal fair cost,
		determine the signature $U_1,\ldots,U_k$ and check whether: $${G \models \psi(c_1,\ldots,c_k,U_1,\ldots,U_k)}.$$
\end{itemize}
Similarly to Theorem~\ref{thm:FPTneighbordiversity}, the algorithm is correct. Moreover, we do only bounded number (Proposition~\ref{prop:extShapeCount})
of \MSOo model checking, so the whole algorithm runs in \FPT time.
\end{proof}

	\section{
Open problems}

The main open problem is whether the bound in Theorems~\ref{thm:edge_deletion_hardness} and~\ref{thm:hardvertex} can be improved to $f(|\psi|,k)n^{\smallo{k/\log k}}$ or even to $f(|\psi|,k)n^{\smallo{k}}$.

The authors would like to thank Martin Koutecký and Petr Hliněný for helpful discussions.

\bibliographystyle{siam}
	\bibliography{src/lit}
%

\end{document}